\documentclass[11pt,letterpaper]{article}
\usepackage{amssymb,amsthm,epsfig}
\usepackage{graphpap}
\usepackage{color}
\usepackage{graphicx}
\usepackage{epstopdf}
\usepackage{pdflscape}
\usepackage{subfig}
\usepackage{epsfig}
\usepackage[cmex10]{amsmath}
\usepackage{amsfonts}
\usepackage{bm}
\usepackage{multirow}
\usepackage{setspace}
\usepackage{times}
\usepackage[square,sort,comma,numbers]{natbib}
\usepackage[ruled,vlined,linesnumbered,commentsnumbered]{algorithm2e}
\usepackage{fullpage}
\usepackage{url}
\usepackage[hmargin=3cm,vmargin=3.5cm]{geometry}

\newtheorem{theorem}{Theorem}
\newtheorem{lemma}{Lemma}

\newtheorem{constraint}{Constraint}
\newtheorem{problem}{Problem}

\newcommand{\remove}[1]{}

\begin{document}

\title{Identifying Correlated Heavy-Hitters in a Two-Dimensional Data Stream
\footnote{A preliminary version of the paper ``Identifying Correlated Heavy-Hitters on a Two-Dimensional Data Stream" was accepted at the Proceedings of the 28th IEEE International Performance Computing and Communications Conference (IPCCC) 2009. }
\footnote{The authors were supported in part by the National Science Foundation through grants NSF CNS-0834743 and CNS-0831903.}
}

\author{
Bibudh Lahiri\thanks{
Impetus Technologies, Los Gatos, CA 95032, USA.
Email: bibudhlahiri@gmail.com
}
\and 
Arko Provo Mukherjee\thanks{
Department of Electrical and Computer Engineering, Iowa State University.
Email: arko@iastate.edu
}
\and 
Srikanta Tirthapura\thanks{
Department of Electrical and Computer Engineering, Iowa State University.
Email: snt@iastate.edu
}
}

\maketitle

\begin{abstract}
  We consider online mining of correlated heavy-hitters from a data
  stream. Given a stream of two-dimensional data, a correlated
  aggregate query first extracts a substream by applying a predicate
  along a primary dimension, and then computes an aggregate along a
  secondary dimension. Prior work on identifying heavy-hitters in
  streams has almost exclusively focused on identifying heavy-hitters
  on a single dimensional stream, and these yield little insight into
  the properties of heavy-hitters along other dimensions. In typical
  applications however, an analyst is interested not only in
  identifying heavy-hitters, but also in understanding further
  properties such as: what other items appear frequently along with a
  heavy-hitter, or what is the frequency distribution of items that
  appear along with the heavy-hitters.

  We consider queries of the following form: ``In a stream $S$ of
  $(x,y)$ tuples, on the substream $H$ of all $x$ values that are
  heavy-hitters, maintain those $y$ values that occur frequently with
  the $x$ values in $H$''. We call this problem as Correlated
  Heavy-Hitters (CHH). We formulate an approximate formulation of CHH
  identification, and present an algorithm for tracking CHHs on a data
  stream. The algorithm is easy to implement and uses workspace which
  is orders of magnitude smaller than the stream itself. We present
  provable guarantees on the maximum error, as well as detailed
  experimental results that demonstrate the space-accuracy trade-off.

\end{abstract}

%---------------------
\section{Introduction}
\label{sec:introduction}
%---------------------
Correlated aggregates \cite{ADGKMS03,GKS01,CTX09} reveal interesting
interactions among different attributes of a multi-dimensional
dataset. They are useful in finding an aggregate on an 
attribute over a subset of the data, where the subset
is defined by a selection predicate on a different attribute of the
data. On stored data, a correlated aggregate can be computed
by considering one dimension at a time, using multiple passes through
the data. However, for dynamic streaming data, we often do not have
the luxury of making multiple passes over the data, and moreover, 
the data may be too large to store and it is desirable to 
have an algorithm that works
in a single pass through the data. Sometimes, even the substream
derived by applying the query predicate along the primary dimension
can be too large to store, let alone the whole dataset.

\remove{
We consider the identification of a) correlated heavy-hitters (CHHs)
and b) computation of the number of distinct items (also known as the
``$0^{th}$ frequency moment'' or $F_0$) from massive data streams.
}

We first define the notion of a heavy-hitter on a data stream (this is
considered in prior work, such as~\cite{MM02,MG82,CCF04,CM2005}), and
then define our notion of correlated heavy-hitters. Given a sequence
of single-dimensional records $(a_1,a_ 2,\ldots,a_N)$, where $a_i \in
\{1,\ldots,m\}$, the frequency of an item $i$ is defined as
$|\{a_j|a_j = i\}|$. Given a user-input threshold $\phi \in (0,1)$,
any data item $i$ whose frequency is at least ${\phi}N$ is termed as a
$\phi$-heavy-hitter. We first consider the following problem of exact
identification of CHHs.

%---------------
\begin{problem}
\label{prob:exact-CHH}
{\bf Exact Identification of Correlated Heavy Hitters.}
Given a data stream $S$ of
$(x, y)$ tuples of length $N$ ($x$ and $y$ will henceforth be
referred to as the ``primary'' and the ``secondary'' dimensions,
respectively), and two user-defined thresholds ${\phi}_1$ and 
${\phi}_2$, where $0 < {\phi}_1 < 1$ and $0 < {\phi}_2 < 1$, identify 
all $(d,s)$ tuples such that:

$$
f_d = |\{(x, y) \in S:(x = d)\}| > {{\phi}_1}N
$$ 
and
$$
f_{d,s} = |\{(x, y) \in S:(x = d) \land (y = s)\}| > {{\phi}_2}f_d
$$
\end{problem}
%---------------

The above aggregate can be understood as follows.  The elements $d$
are heavy-hitters in the traditional sense, on the stream formed by
projecting along the primary dimension. For each heavy-hitter $d$
along the primary dimension, there is logically a (uni-dimensional)
substream $S_d$, consisting of all values along the secondary
dimension, where the primary dimension equals $d$. We require the
tracking of all tuples $(d,s)$ such that $s$ is a heavy-hitter in
$S_d$.

Many stream mining and monitoring problems on two-dimensional streams
need the CHH aggregate, and cannot be answered by independent
aggregation along single dimensions. For example, consider a network
monitoring application, where a stream of (destination IP address,
source IP address) pairs is being observed. The network monitor maybe
interested not only in tracking those destination IP addresses that
receive a large fraction of traffic (heavy-hitter destinations), but
also in tracking those source IP addresses that send a large volume of
traffic to these heavy-hitter destinations. This cannot be done by
independently tracking heavy-hitters along the primary and the
secondary dimensions. Note that in this application, we are interested
not only in the identity of the heavy-hitters, but also additional
information on the substream induced by the heavy-hitters.

In another example, in a stream of (server IP address, port number)
tuples, identifying the heavy-hitter server IP addresses will tell us
which servers are popular, and identifying frequent port numbers
(independently) will tell us which applications are popular; but a
network manager maybe interested in knowing which applications are
popular among the heavily loaded servers, which can be retrieved using
a CHH query. Such correlation queries are used for network
optimization and anomaly detection \cite{Cullingford09}.

Another application is the recommendation system of a typical online
shopping site, which shows a buyer a list of the items frequently
bought with the ones she has decided to buy. Our algorithm can
optimize the performance of such a system by parsing the transaction
logs and identifying the items that were bought commonly with the
frequently purchased items. If such information is stored in a cache
with a small lookup time, then for most buyers, the recommendation
system can save the time to perform a query on the disk-resident data.

Similar to the above examples, in many stream monitoring applications,
it is important to track the heavy-hitters in the stream, but this
monitoring should go beyond simple identification of heavy-hitters, or
tracking their frequencies, as is considered in most prior
formulations of heavy-hitter tracking such as
\cite{CM03,MM02,MG82,CCF04,EV02}. In this work we initiate the study
of tracking additional properties of heavy-hitters by considering
tracking of correlated heavy hitters.

%None of the above
%mining tasks on streams would have been possible using a traditional
%heavy-hitter algorithm on streams, since they do not consider
%correlations across multiple dimensions.

%---------------------------
\subsection{Approximate CHH}
\label{sec:problem}
%---------------------------

It is easy to prove that exact identification of heavy-hitters in a
single dimension is impossible using limited space, and one pass
through the input. Hence, the CHH problem is also impossible to solve
in limited space, using a single pass through the input. Due to this,
we consider the following approximate version of the problem. We
introduce additional approximation parameters, ${\epsilon}_1$ and
${\epsilon}_2$ ($0 < {\epsilon}_1 \le \frac{{\phi}_1}{2}$, $0 <
{\epsilon}_2 < {\phi}_2$), which stand for the approximation errors
along the primary and the secondary dimensions, respectively. We seek
an algorithm that provides the following guarantees.

%---------------
\begin{problem}
\label{prob:approx-CHH}
{\bf Approximate Identification of Correlated Heavy-Hitters.}  Given a
data stream $S$ of $(d,s)$ tuples of length $N$, thresholds ${\phi}_1$
and ${\phi}_2$:

\begin{enumerate}
\item
Report any value $d$ such that $f_d > {{\phi}_1}N$ as a
heavy-hitter along the primary dimension.

\item
No value $d$ such that $f_d < ({\phi}_1 - {\epsilon}_1)N$, should be reported as
a heavy-hitter along the primary dimension.

\item
For any value $d$ reported above, report any value $s$ along the secondary
dimension such that $f_{d,s} > {{\phi}_2}f_d$ as a CHH.

\item
For any value $d$ reported above, no value $s$ along the secondary
dimension such that $f_{d,s} < ({{\phi}_2} - {{\epsilon}_2})f_d$
should be reported as a CHH occurring alongwith $d$.
\end{enumerate}
\end{problem}
%---------------

With this problem formulation, false positives are possible, but false
negatives are not. In other words, if a pair $(d,s)$ is a CHH
according to the definition in Problem~\ref{prob:exact-CHH}, then it
is a CHH according to the definition in Problem~\ref{prob:approx-CHH},
and will be returned by the algorithm. But an algorithm for Problem
\ref{prob:approx-CHH} may return a pair $(d,s)$ that are not exact
CHHs, but whose frequencies are close to the required thresholds.

%-------------------------
\subsection{Contributions}
%-------------------------
Our contributions are as follows.

\begin{itemize}
\item 
We formulate exact and approximate versions of the problem of
identifying CHHs in a multidimensional data stream, and present a
small-space approximation algorithm for identifying approximate CHHs
in a single pass. Prior literature on correlated aggregates have
mostly focused on the correlated sum, and these techniques are not
applicable for CHH. Our algorithm for approximate CHH identification
is based on a nested application of the Misra-Gries algorithm~\cite{MG82}.

%We address the question of analyzing heavy-hitters along other dimensions for a 
%multidimensional stream.
%We introduce two problems for this: identifying items that occur frequently along 
%with the heavy-hitters, and computing the number of distinct items along some other dimension that occur with the 
%heavy-hitters. We prove theoretically that the second problem has no small-space solution.

\item
We provide a provable guarantee on the approximation error. We show
that there are no false negatives, and the error in the false
positives is controlled. When greater memory is available, this error
can be reduced. The space taken by the algorithm as well as the
approximation error of the algorithm depend on the sizes of two
different data structures within the algorithm. The total space taken
by the sketch is minimized through solving a constrained optimization
problem that minimizes the total space taken subject to providing the
user-desired error guarantees.

\item
We present results from our simulations on a) a stream of more than
1.4 billion (50 GB trace) anonymized packet headers from an OC48 link
(collected by CAIDA \cite{CAIDA}), and b) a sample of 240 million
2-grams extracted from English fiction books \cite{ngram}.  We
compared the performance of our small-space algorithm with a slow, but
exact algorithm that goes through the input data in multiple
passes. Our experiments revealed that even with a space budget of a
few megabytes, the average error of our algorithm was very small,
showing that it is viable in practice.
\end{itemize}

Along each dimension our algorithm maintains frequency estimates of
mostly those values (or pairs of values) that occur frequently. For
example, in a stream of (destination IP, source IP) tuples, 
for every destination that sends a significant fraction of
traffic on a link, we maintain mostly the sources that occur
frequently along with this destination. Note that the set of
heavy-hitters along the primary dimension can change as the stream
elements arrive, and this influences the set of CHHs along the
secondary dimension. For example, if an erstwhile heavy-hitter
destination $d$ no longer qualifies as a heavy-hitter with increase in
$N$ (and hence gets rejected from the sketch), then a source $s$
occurring with $d$ should also be discarded from the sketch. This
interplay between different dimensions has to be handled carefully
during algorithm design.\\

%---------------------
\noindent{\bf Roadmap:} The rest of this paper is organized as
follows. We present related work in Section~\ref{sec:relatedwork}.
In Section~\ref{sec:algo} we present the algorithm description,
followed by the proof of correctness in Section~\ref{sec:correctness},
and the analysis of the space complexity in
Section~\ref{sec:analysis}. We present experimental results in
Section~\ref{sec:simulation}.

%---------------------
\section{Related Work}
\label{sec:relatedwork}
%---------------------

In the data streaming literature, there is a significant body of work
on correlated aggregates (\cite{ADGKMS03,GKS01,CTX09}), as well as on
the identification of heavy hitters
(\cite{MM02,MG82,CCF04,CM2005}). See \cite{CH09} for a recent overview
of work on heavy-hitter identification. None of these works consider
correlated heavy-hitters.

Estan {\em et al.} \cite{ESV03} and Zhang {\em et al.} \cite{ZSSDL04}
have independently studied the problem of identifying heavy-hitters
from multi-dimensional packet streams, but they both define a
multidimensional tuple as a heavy-hitter if it occurs more than
${\phi}N$ times in the stream, $N$ being the stream size -- the
interplay across different dimensions is not considered.

There is significant prior work on correlated aggregate computation
that we now describe. The problems considered in the literature
usually take the following form. On a stream of two dimensional data
items $(x,y)$ the query asks to first apply a selection predicate
along the $x$ dimension, of the form $x \ge c$ or $x < c$ (for a value
$c$ provided at query time), followed by an aggregation along the $y$
dimension. The difference when compared with this formulation is that
in our case, the selection predicate along the $x$ dimension is one
that involves frequencies and heavy-hitters, rather than a simple
comparison.

Gehrke {\em et al}~\cite{GKS01} addressed correlated aggregates where
the aggregate along the primary dimension was an extremum (min or max)
or the average, and the aggregate along the secondary dimension was
sum or count. For example, given a stream $S$ of $(x,y)$ tuples, their
algorithm could approximately answer queries of the following form:
``Return the sum of $y$-values from $S$ where the corresponding $x$
values are greater than a threshold $\alpha$.'' They describe a data
structure called {\em adaptive histograms}, but these did not come
with provable guarantees on performance. 
Ananthakrishna {\em et al} \cite{ADGKMS03} presented algorithms with
provable error bounds for correlated sum and count. Their solution was
based on the quantile summary of \cite{GK01}.  With this technique,
heavy-hitter queries cannot be used as the aggregate along the primary
dimension since they cannot be computed on a stream using limited
space. Cormode, Tirthapura, and Xu~\cite{CTX09} presented algorithms
for maintaining the more general case of {\em time-decayed} correlated
aggregates, where the stream elements were weighted based on the time
of arrival. This work also addressed the ``sum'' aggregate, and the
methods are not directly applicable to heavy-hitters. Other work in this 
direction includes~\cite{BT07,XTB08}. Tirthapura and
Woodruff~\cite{TW12} present a general method that reduces the
correlated estimation of an aggregate to the streaming computation of the
aggregate, for functions that admit sketches of a particular
structure. These techniques only apply to selection predicates of
the form $x > c$ or $x < c$, and do not apply to heavy-hitters, as we
consider here.

The heavy-hitters literature has usually focused on the following problem.
Given a sequence of elements $A = (a_1,a_2,\ldots,a_N)$ and a user-input
threshold $\phi \in (0,1)$, find data items that occur more than $\phi N$
times in $A$. Misra and Gries~\cite{MG82} presented a deterministic algorithm
for this problem, with space complexity being $O(\frac{1}{\phi})$, time
complexity for updating the sketch with the arrival of each element being 
$O(\log{\frac{1}{\phi}})$, and query time complexity
being $O(\frac{1}{\phi})$. For exact identification of heavy-hitters, their
algorithm works in two passes. For approximate heavy-hitters, their algorithm used
only one pass through the sequence, and had the following approximation
guarantee. Assume user-input threshold $\phi$ and approximation error
$\epsilon < \phi$. Note that for an online algorithm, $N$ is the number of
elements received so far.

\begin{itemize}
\item All items whose frequencies exceed ${\phi}N$ are output.
i.e. there are no false negatives.
\item No item with frequency less than $(\phi-\epsilon)N$ is output. 
\end{itemize}

Demaine {\em et al} \cite{DLM02} and Karp {\em et al} \cite{KSP03}
improved the sketch update time per element of the Misra-Gries
algorithm from $O(\log{\frac{1}{\phi}})$ to $O(1)$, using an advanced
data structure combining a hashtable, a linked list and a set of
doubly-linked lists. Manku and Motwani \cite{MM02} presented a
deterministic ``Lossy Counting'' algorithm that offered the same
approximation guarantees as the one-pass approximate Misra-Gries
algorithm; but their algorithm required
$O(\frac{1}{\epsilon}\log{({\epsilon}N)})$ space in the worst
case. For our problem, we chose to extend the Misra-Gries algorithm as
it takes asymptotically less space than \cite{MM02}.

\remove{Cormode and Muthukrishnan \cite{CM2005}
came up with a probabilistic ``Count-Min Sketch'' that provided the
following guarantees: for a stream which allows only insertion of
items and not deletion, it identified every item with a frequency
more than ${\phi}N$; and with probability at least $1-\theta$ (for
some input parameter $\theta \in (0,1)$), it did not output any item
whose frequency is less than $(\phi-\epsilon)N$. It required
$O(\frac{1}{\epsilon}\log{\frac{N}{\theta}})$ space and
$O(\log{\frac{N}{\theta}})$ update time per item; and
$O(\log{\frac{N}{\theta}})$ pairwise independent hash functions to maintain the data structure. 
Clearly, the Misra-Gries algorithm takes asymptotically less space than either of these two, and can be 
implemented by simpler data structures.

Our problem is different from the well-known problem of association
rule mining. Following the terminology of \cite{AIS93,AS94}, if we
consider the problem of mining all association rules of the form $X
\Rightarrow Y$ with {\em support} $\phi$ and {\em confidence}
$\psi$, where $X$ and $Y$ are 1-itemsets, then, {\em both} $X$ and
$Y$ would have to be present in at least $\phi$ fraction of the
stream elements, which is clearly not our requirement.}

%Moreover, the
%existing association rule mining algorithms (even \cite{CHNW96})
%typically make multiple passes over the data, which is not
%acceptable in a streaming scenario.

%--------------------------------
\section{Algorithm and Analysis}
%--------------------------------

\subsection{Intuition and Algorithm Description}
\label{sec:algo}

Our algorithm is based on a nested application of an algorithm for
identifying frequent items from an one-dimensional stream, due to
Misra and Gries \cite{MG82}. We first describe the Misra-Gries
algorithm (henceforth called the MG algorithm). Suppose we are given
an input stream $a_1,a_2,\ldots$, and an error threshold $\epsilon, 0
< \epsilon < 1$.  The algorithm maintains a data structure
$\mathcal{D}$ that contains at most $\frac{1}{\epsilon}$ (key, count)
pairs.  On receiving an item $a_i$, it is first checked if a tuple
$(a_i, \cdot)$ already exists in $\mathcal{D}$.  If it does, $a_i$'s
count is incremented by 1; otherwise, the pair $(a_i,1)$ is added to
$\mathcal{D}$.  Now, if adding a new pair to $\mathcal{D}$ makes
$|\mathcal{D}|$ exceed $\frac{1}{\epsilon}$, then for each (key,
count) pair in $\mathcal{D}$, the count is decremented by one; and any
key whose count falls to zero is discarded. This ensures at least the
key which was most recently added (with a count of one) would get
discarded, so the size of $\mathcal{D}$, after processing all pairs,
would come down to $\frac{1}{\epsilon}$ or less. Thus, the space
requirement of this algorithm is $O(\frac{1}{\epsilon})$. The data
structure $\mathcal{D}$ can be implemented using hashtables or
height-balanced binary search trees. At the end of one pass through
the data, the MG algorithm maintains the frequencies of keys in the
stream with an error of no more than $\epsilon n$, where $n$ is the
size of the stream. The MG algorithm can be used in exact
identification of heavy hitters from a data stream using two passes
through the data.

In the scenario of limited memory, the MG algorithm can be used to
solve problem~\ref{prob:exact-CHH} in three passes through the data,
as follows. We first describe a four pass algorithm. In the first two
passes, heavy-hitters along the primary dimension are identified,
using memory $O(1/{\phi}_1)$. Note that this is asymptotically the
minimum possible memory requirement of any algorithm for identifying
heavy-hitters, since the size of output can be
$\Omega\left(\frac{1}{{\phi}_1}\right)$. In the next two passes,
heavy-hitters along the secondary dimension are identified for each
heavy-hitter along the primary dimension. This takes space
$O\left(\frac{1}{{\phi}_2}\right)$ for each heavy-hitter along the
primary dimension. The total space cost is
$O\left(\frac{1}{{{\phi}_1}{{\phi}_2}}\right)$, which is optimal,
since the output could be
$\Omega\left(\frac{1}{{{\phi}_1}{{\phi}_2}}\right)$ elements. The
above algorithm can be converted into a {\em three} pass exact
algorithm by combining the second and third passes.

The high-level idea behind our single-pass algorithm for
Problem~\ref{prob:approx-CHH} is as follows.  The MG algorithm for an
one-dimensional stream, can be viewed as maintaining a small space
``sketch'' of data that (approximately) maintains the frequencies of
each distinct item $d$ along the primary dimension; of course, these
frequency estimates are useful only for items that have very high
frequencies. For each distinct item $d$ along the primary dimension,
apart from maintaining its frequency estimate ${\hat{f}}_d$, our
algorithm maintains an embedded MG sketch of the substream $S_d$
induced by $d$, i.e. $S_d = \{(x,y)| ((x,y) \in S) \land (x =
d)\}$. The embedded sketch is a set of tuples of the form $(s,
{\hat{f}}_{d,s})$, where $s$ is an item that occurs in $S_d$, and
${\hat{f}}_{d,s}$ is an estimate of the frequency of the pair $(d,s)$
in $S$ (or equivalently, the frequency of $s$ in $S_d$). While the
actions on ${\hat{f}}_d$ (increment, decrement, discard) depend on how
$d$ and the other items appear in $S$, the actions on
${\hat{f}}_{d,s}$ depend on the items appearing in $S_d$. Further, the
sizes of the tables that are maintained have an important effect on
both the correctness and the space complexity of the algorithm.

We now present a more detailed description. The algorithm maintains a
table $H$, which is a set of tuples $(d,{\hat{f}}_d, H_d)$, where $d$
is a value along the primary dimension, ${\hat{f}}_d$ is the estimated
frequency of $d$ in the stream, and $H_d$ is another table that stores
the values of the secondary attribute that occur with $d$.  $H_d$
stores its content in the form of (key, count) pairs, where the keys
are values ($s$) along the secondary attribute and the counts are the
frequencies of $s$ in $S_d$, denoted as ${\hat{f}}_{d,s}$, alongwith
$d$.

The maximum number of tuples in $H$ is $s_1$, and the maximum number
of tuples in each $H_d$ is $s_2$. The values of $s_1$ and $s_2$ depend
on the parameters $\phi_1,\phi_2, \epsilon_1,\epsilon_2$, and are
decided at the start of the algorithm. Since $s_1$ and $s_2$ effect
the space complexity of the algorithm, as well as the correctness
guarantees provided by it, their values are set based on an
optimization procedure, as described in Section \ref{sec:analysis}.

The formal description is presented in Algorithms \ref{algo:init},
\ref{algo:update} and \ref{algo:report}. Before a stream element is
received, Algorithm \ref{algo:init} {\sf Sketch-Initialize} is invoked
to initialize the data structures. Algorithm \ref{algo:update} {\sf
Sketch-Update} is invoked to update the data structure as each stream
tuple $(x,y)$ arrives. Algorithm \ref{algo:report} {\sf Report-CHH}
is used to answer queries when a user asks for the CHHs in the stream
so far.

On receiving an element $(x,y)$ of the stream, the following three
scenarios may arise. We explain the action taken in each.
\begin{enumerate}
\item
If $x$ is present in $H$, and $y$ is present in $H_x$, then both
${\hat{f}}_{x}$ and ${\hat{f}}_{x,y}$ are incremented.

\item
If $x$ is present in $H$, but $y$ is not in $H_x$, then $y$ is added
to $H_x$ with a count of 1. If this addition causes $|H_x|$ to exceed
its space budget $s_2$, then for each (key, count) pair in $H_x$, the
count is decremented by 1 (similar to the MG algorithm). If the count
of any key falls to zero, the key is dropped from $H_x$. Note that
after this operation, the size of $H_x$ will be at most $s_2$.

\item
If $x$ is not present in $H$, then an entry is created for $x$ in $H$
by setting ${\hat{f}}_{x}$ to 1, and by initializing $H_x$ with the
pair $(y,1)$. If adding this entry causes $|H|$ to exceed $s_1$, then
for each $d \in H$, $f_d$ is decremented by $1$.  If the decrement
causes ${\hat{f}}_d$ to be zero, then we simply discard the entry for
$d$ from $H$.

Otherwise, when $f_d$ is decremented, the algorithm keeps the sum of
the $\hat{f_{d,s}}$ counts within $H_d$ equal to $f_d$; the detailed
correctness is proved in Section \ref{sec:analysis}. To achieve this,
an arbitrary key $s$ is selected from $H_d$ such that such that
${\hat{f}}_{d,s} > 0$, and ${\hat{f}}_{d,s}$ is decremented by $1$.
If ${\hat{f}}_{d,s}$ falls to zero, $s$ is discarded from $H_d$.
\end{enumerate}

%-------------------
\begin{algorithm}
\caption{{\sf Sketch-Initialize}$({\phi}_1,{\phi}_2,{\epsilon}_1,{\epsilon}_2)$}
\label{algo:init}
\KwIn{Threshold for primary dimension ${\phi}_1$; Threshold for secondary dimension ${\phi}_2$; 
      Tolerance for primary dimension ${\epsilon}_1$;
      Tolerance for secondary dimension ${\epsilon}_2$}
\BlankLine
$H \gets \Phi$

Set $s_1$ and $s_2$ as described in Section \ref{sec:analysis}.
%$s_1 \gets \max\left(\frac{1}{{\epsilon}_1},\frac{8}{{\phi}_1{\epsilon}_2}\right)$;
%$s_2 \gets \frac{1}{{\epsilon}_2 - \frac{4}{s_1{\phi}_1}}$
\end{algorithm}

%--------------------------
\begin{algorithm}
\caption{{\sf Sketch-Update}$(x, y)$}
\label{algo:update}
\KwIn{Element along primary dimension $x$; Element along secondary dimension $y$}
\BlankLine
\eIf{$x \in H$}{
	${\hat{f}}_{x} \gets {\hat{f}}_{x} + 1$\;
	\eIf{$y \in H_x$}{
	   \tcc{Both $x$ and $y$ are present}
	   Increment ${\hat{f}}_{x,y}$ in $H_x$ by 1\;
	}
	%the first else block begins
	{
	   \tcc{$x \in H$, but $y \not\in H_x$}
	   Add the tuple $(y,1)$ to $H_x$\;
	   \If {$|H_x| > s_2$}
	   {
	      \ForEach{$(s,{\hat{f}}_{d,s}) \in H_x$}
	      {
				${\hat{f}}_{d,s} \gets {\hat{f}}_{d,s} - 1$\;\nllabel{line:decrement_f_d_s}
				\If{${\hat{f}}_{d,s} = 0$}
				{
				discard $(s,{\hat{f}}_{d,s})$ from $H_x$\;
				}
			}
		}
	} %end first else block
}
{
	\tcc{Neither of $x$ or $y$ is present}
	$H_x \gets \Phi$;
	Add $(y,1)$ to $H_x$; ${\hat{f}}_{x} \gets 1$\;

	\If {$|H| > s_1$}{
		\ForEach{$d \in H$}
		{
			${\hat{f}}_d \gets {\hat{f}}_d - 1$\;

			\If {there exists $s$ such that ${\hat{f}}_{d,s} > 0$}
			{
			  Choose an arbitrary 
                          $(s,{\hat{f}}_{d,s}) \in H_d$ such that
                          ${\hat{f}}_{d,s} > 0$\; \nllabel{line:arbitrary}

			  ${\hat{f}}_{d,s} \gets {\hat{f}}_{d,s} - 1$\;\nllabel{line:decrement_f_d_s_arbitrary}

			  \If {${\hat{f}}_{d,s} = 0$}
			  {
			    discard $(s,{\hat{f}}_{d,s})$ from $H_d$;
			  }
			}
			\If {${\hat{f}}_d = 0$}
			{
 			  Discard $(d, H_d)$ from $H$\;\nllabel{line:discard_d_H_d}
			}
		}
	}
}
\end{algorithm}

%----------------------------
\begin{algorithm}
\caption{{\sf Report-CHH}($N$)}
\label{algo:report}
\KwIn{Size of the stream $N$}
\BlankLine
\ForEach{$d \in H$}
{
 \If{${\hat{f}}_d \ge ({\phi}_1- \frac{1}{s_1})N$}
 {
   Report $d$ as a frequent value of the primary attribute\;
   \ForEach{$(s,{\hat{f}}_{d,s}) \in H_d$}
   {
     \If{${\hat{f}}_{d,s} \ge ({\phi}_2- \frac{1}{s_2}){\hat{f}}_d - \frac{N}{s_1}$}
     {
       Report $s$ as a CHH occurring with $d$\;
     }
   }
 }
}
\end{algorithm}

%-----------------------------------
\subsection{Algorithm Correctness}
\label{sec:correctness}

\newcommand{\fd}{f_d}
\newcommand{\fds}{f_{d,s}}
\newcommand{\estfd}{\hat{f}_d}
\newcommand{\estfds}{\hat{f}_{d,s}}
\newcommand{\sumfds}{\Sigma_d}
%---------------------

In this section, we show the correctness of the algorithm, subject to 
the following constraints on $s_1$ and $s_2$. In Section
\ref{sec:analysis}, we assign values to $s_1$ and $s_2$ in such a
manner that the space taken by the data structure is minimized.

\begin{constraint}
\label{const:s1}
$$\frac{1}{s_1} \le {\epsilon}_1$$
\end{constraint}

\begin{constraint}
\label{const:s1s2}
$$\frac{1}{s_2} + \frac{1+{\phi}_2}{s_1({\phi}_1-{\epsilon}_1)} \le {\epsilon}_2$$
\end{constraint}

Consider the state of the data structure after a stream $S$ of length
$N$ has been observed. Consider a value $d$ of the primary attribute,
and $s$ of the secondary attribute. Let $f_d$ and $f_{d,s}$ be defined
as in Section \ref{sec:introduction}. Our analysis focuses on the
values of variables $\estfd$ and $\estfds$, which are updated in
Algorithms \ref{algo:update} and used in Algorithm
\ref{algo:report}. For convenience, if $d$ is not present in $H$ then
we define $\estfd=0$.  Similarly, if $d$ is not present in $H$, or if
$(d,s)$ is not present in $H_d$, then we define $\estfds=0$.

%-------------------------------
\begin{lemma}
\label{lemm:dstIP_epsilon} 
$$\estfd \ge \fd - \frac{N}{s_1}$$
\end{lemma}
%-------------------------------

\begin{proof}
The total number of increments in the $s_1$ counters that keep track
of the counts of the different values of the primary attribute is $N$.
Each time there is a decrement to $\estfd$ (in Line 20 of Algorithm
\ref{algo:update}), $s_1+1$ different counters are decremented.
The total number of decrements, however, cannot be more than the total
number of increments, and hence is at most $N$. So the number of times
the block of lines 19-31 in Algorithm~\ref{algo:update} gets executed
is at most $\frac{N}{s_1 + 1} < \frac{N}{s_1}$. We also know that
$\estfd$ is incremented exactly $\fd$ times, hence the final value of
$\estfd$ is greater than $f_d-\frac{N}{s_1}$. 
\end{proof}

%-----------------
\begin{lemma}
\label{lem:dstIP} 
Assume that Constraint~\ref{const:s1} is true.
If $\fd > {\phi}_1 N$, then $d$ is reported by
Algorithm~\ref{algo:report}  as a frequent item.
Further, if $\fd < (\phi_1 - \epsilon_1)N$, then 
$d$ is not reported as a frequent item.
\end{lemma}

\begin{proof}
Suppose $\fd \ge \phi_1 N$. From Lemma~\ref{lemm:dstIP_epsilon}, 
$\estfd \ge \fd - {\epsilon}_1 N \ge {{\phi}_1}N -  {\epsilon}_1 N$.
Hence Algorithm~\ref{algo:report} will report $d$ (see Lines 2 and 3).
Next, suppose that $\fd < (\phi_1 - \epsilon_1)N$. Since $\estfd \le \fd$, 
Algorithm~\ref{algo:report} will not report $d$ as a frequent item.
\end{proof}

%-------------------------

\begin{lemma} 
\label{lemm:invariant} 
$$
\sum_{(s,\cdot) \in H_d} \estfds \le \estfd
$$
\end{lemma}

\begin{proof}
Let $\sumfds = \sum_{(s,\cdot) \in H_d} \estfds$.
Let $C(n)$ denote the condition $\sumfds \le \estfd$ after $n$
stream elements have been observed.
We prove $C(n)$ by induction on $n$.  
The base case is when $n=0$, and in this case,
$\estfds=\estfd=0$ for all $d,s$, and $C(0)$ is trivially true.
For the inductive step, assume that $C(k)$ is true, for $k \ge 0$. 
Consider a new element that arrives, say $(x,y)$, and consider 
Algorithm \ref{algo:update} applied 
on this element. We consider four possible cases.

(I) If $x=d$, and $d \in H$, then $\estfd$ is incremented by $1$, and
it can be verified (Lines 3-11) that $\sumfds$ increases 
by at most $1$ (and may even decrease). Thus $C(k+1)$ is true.

(II) If $x=d$, and $d \not\in H$, then initially, $\estfd$ and $\sumfds$
are both 1 (line 17). If $|H| \le s_1$, then both $\estfd$ and
$\sumfds$ remain 1, and $C(k+1)$ is true.
Suppose $|H| > s_1$, then both $\estfd$ and $\sumfds$ will go down to
$0$, since $H_d$ will be discarded from $H$. Thus $C(k+1)$ is true.

(III) If $x \neq d$, and $x \in H$, then neither $\estfd$ nor 
$\sumfds$ change. 

(IV) Finally, if $x \neq d$ and $x \not\in H$, then it is possible that
$\estfd$ is decremented (line 20). In this case, if $\sumfds > 0$,
then $\sumfds$ is also decremented (line 22), and $C(k+1)$ is
satisfied. If $\sumfds = 0$, then $C(k+1)$ is trivially satisfied
since $\estfd \ge 0$.
\end{proof}

%----------------
\begin{lemma}
\label{lemm:srcIP_epsilon} 
Subject to Constraint~\ref{const:s1},
$\estfds \ge \fds - \epsilon_2 \fd - \epsilon_1 N$.
\end{lemma}

\begin{proof}
Note that each time the tuple $(d,s)$ occurs in the stream,
$\estfds$ is incremented in Algorithm \ref{algo:update}. But
$\estfds$ can be less than $\fds$ because of decrements in Lines
\ref{line:decrement_f_d_s} or \ref{line:decrement_f_d_s_arbitrary}
in Algorithm \ref{algo:update}. We consider these two cases separately.

Let $\sumfds = \sum_{(s,\cdot) \in H_d} \estfds$. 
For decrements in Line 9, we observe that each time this line is
executed, $\sumfds$ reduces by $s_2+1$. From Lemma
\ref{lemm:invariant}, we know that $\sumfds \le \estfd \le \fd$.
Thus the total number of times $\estfds$ is decremented due to Line 9
is no more than $\frac{\fd}{s_2+1}$. From Constraint \ref{const:s1s2}, we know
$\frac{1}{s_2} < \epsilon_2$, and $\frac{\fd}{s_2+1} < \epsilon_2 \fd$.

For decrements in Line 23, we observe that $\estfds$ is decremented in
Line 23 no more than the number of decrements to $\estfd$, which was
bounded by $\frac{N}{s_1}$ in Lemma \ref{lemm:dstIP_epsilon}. From
Constraint \ref{const:s1}, this is no more than $\epsilon_1 N$.
\end{proof}

%---------------------

\begin{lemma}
\label{lemm:no-fn-srcIP} 
For any value $d$ that gets reported in line 3 of
Algorithm~\ref{algo:report}, any value $s$ of the secondary attribute
that occurs with $d$ such that $f_{d,s} > \phi_2 f_d$, will be
identified by line 6 of Algorithm~\ref{algo:report} as a CHH occurring
alongwith $d$.
\end{lemma}

\begin{proof}
From Lemma \ref{lemm:srcIP_epsilon}, 
\begin{eqnarray*}
\estfds & \ge & f_{d,s} - {{\epsilon}_2}f_d - {{\epsilon}_1}N      \\
        & >   & {{\phi}_2}f_d - {{\epsilon}_2}f_d - {{\epsilon}_1}N  \\
        & =   & ({\phi}_2 - {\epsilon}_2)f_d - {\epsilon}_1N       \\
        & \ge & ({\phi}_2 - {\epsilon}_2){\hat{f}}_d - {\epsilon}_1N
\end{eqnarray*}

where we have used $\fd \ge \estfd$.
The lemma follows since $({\phi}_2 - {\epsilon}_2){\hat{f}}_d - {\epsilon}_1N$ is the threshold
used in line 5 of Algorithm \ref{algo:report} to report a value of the secondary attribute
as a CHH.
\end{proof}

%--------------------------------------
\begin{lemma}
\label{lemm:fp-srcIP} 
Under Constraints~\ref{const:s1} and~\ref{const:s1s2}, 
for any value of $d$ that is reported as a heavy-hitter along the
primary dimension, then for a value $s'$ along the secondary
dimension, if $f_{d,s'} < (\phi_2 - \epsilon_2) \fd$, then the pair
$(d,s')$ will not be reported as a CHH.
\end{lemma}
%--------------------------------------

\begin{proof}

We will prove the contrapositive of the above statement. Consider a
value $s$ such that $(d,s)$ is reported as a CHH. Then, we show that
$\fds \ge (\phi_2 - \epsilon_2) \fd$.

If $(d,s)$ is reported, then it must be true that 
$\estfds \ge (\phi_2-\frac{1}{s_2}) \estfd - \frac{N}{s_1}$ (Algorithm
\ref{algo:report}, line 5).
Using $\fds \ge \estfds$, and $\estfd \ge \fd - \frac{N}{s_1}$, we get:

\begin{eqnarray*}
\fds & \ge & \estfds \\
     & \ge & (\phi_2-\frac{1}{s_2}) \estfd - \frac{N}{s_1} \\
     & \ge & (\phi_2-\frac{1}{s_2}) (\fd - \frac{N}{s_1}) - \frac{N}{s_1} \\
     & =   & (\phi_2 - \frac{1}{s_2}) \fd - \frac{N}{s_1}\left(1+{\phi}_2-\frac{1}{s_2}\right)\\
     & \ge & ({\phi}_2 - \frac{1}{s_2}) \fd - \frac{f_d}{({\phi}_1-{\epsilon}_1)s_1}\left(1+{\phi}_2-\frac{1}{s_2}\right)\\
     &     & \mbox{ (since $d$ gets reported, by Lemma~\ref{lem:dstIP}, $f_d \ge ({\phi}_1-{\epsilon}_1)N \Rightarrow N \le \frac{f_d}{{\phi}_1-{\epsilon}_1}$)}\\
     & =   & \left({\phi}_2 - \frac{1}{s_2} - \frac{1}{({\phi}_1-{\epsilon}_1)s_1}\left(1+{\phi}_2-\frac{1}{s_2}\right)\right)f_d\\
     & \ge & f_d({\phi}_2-{\epsilon}_2) \mbox{(using Constraint~\ref{const:s1s2})}
\end{eqnarray*}
\end{proof}

Lemmas \ref{lemm:fp-srcIP}, \ref{lemm:no-fn-srcIP}, and
\ref{lem:dstIP} together yield the following.
\begin{theorem}
\label{thm:correctness}
If Constraints~\ref{const:s1} and~\ref{const:s1s2} are satisfied, then
Algorithms \ref{algo:init}, \ref{algo:update} and \ref{algo:report}
satisfy all the four requirements of Problem~\ref{prob:approx-CHH}.
\end{theorem}

%---------------------
\subsection{Analysis}
\label{sec:analysis}
%---------------------

We analyze the space complexity of the algorithm.  In Theorem
\ref{thm:correctness}, we showed that the Algorithms
\ref{algo:update} and \ref{algo:report} solve the Approximate CHH
detection problem, as long as constraints \ref{const:s1} and
\ref{const:s1s2} are satisfied. 

{\bf Space Complexity in terms of $s_1$ and $s_2$}. In
our algorithm, we maintain at most $s_2$ counters for each of the (at
most) $s_1$ distinct values of the primary attribute in $H$. Hence,
the size of our sketch is $O(s_1 + s_1s_2) = O(s_1s_2)$. 
We now focus on the following question. {\em What is the setting of
$s_1$ and $s_2$ so that the space complexity of the sketch is
minimized while meeting the constraints required for correctness.?}

%--------------------------------
\begin{lemma}
Let $\alpha= \left(\frac{1+{\phi}_2}{{\phi}_1-{\epsilon}_1}\right)$.
Subject to constraints \ref{const:s1} and \ref{const:s1s2}, the space
of the data structure is minimized by the following settings of $s_1$
and $s_2$. 

\begin{itemize}
\item
If $\epsilon_1 \ge \frac{\epsilon_2}{2\alpha}$, then 
$s_1 = \frac{2\alpha}{\epsilon}$ and 
$s_2 = \frac{2}{\epsilon_2}$.
In this case, the space complexity is $O\left(\frac{1}{(\phi_1-\epsilon_1)\epsilon_2^2}\right)$.

\item
If  $\epsilon_1 < \frac{\epsilon_2}{2\alpha}$, then
$s_1 = \frac{1}{\epsilon_1}$, and 
$s_2 = \frac{1}{\epsilon_2 - \alpha \epsilon_1}$.
In this case, the space complexity is $O(\frac{1}{\epsilon_1 \epsilon_2})$.
\end{itemize}
\end{lemma}

%--------------------------
\begin{proof}
Let $\sigma_1 = \frac{1}{s_1}$, $\sigma_2 = \frac{1}{s_2}$.
The problem is now to maximize $\sigma_1 \sigma_2$. 
Constraints \ref{const:s1} and \ref{const:s1s2} can be rewritten 
as follows.

\begin{itemize}
\item
{\bf Constraint 1:} $\sigma_1 \le \epsilon_1$
\item
{\bf Constraint 2:} $\alpha \sigma_1 + {\sigma}_2 \le {\epsilon}_2$
\end{itemize}

First, we note that any assignment $(\sigma_1, \sigma_2) = (x,y)$ that
maximizes $\sigma_1 \sigma_2$ must be tight on Constraint 2, i.e. 
$\alpha x + y = \epsilon_2$. This can be proved by
contradiction. Suppose not, and $\alpha x + y < \epsilon_2$, and $xy$
is the maximum possible. Now, 
there is a solution $\sigma_1 = x$, and $\sigma_2 = y'$, such that 
$y < y'$, and Constraints 1 and 2 are still satisfied. Further, 
$xy' > xy$, showing that the solution $(x,y)$ is not optimal.

Thus, we have:
\begin{equation}
\sigma_2 = \epsilon_2 - \alpha\sigma_1
\end{equation}

Thus the problem has reduced to:
{\bf Maximize $f(\sigma_1) = \sigma_1 \left(\epsilon_2 - \alpha\sigma_1\right)$ subject to 
$\sigma_1 \le \epsilon_1$.}

Consider
$$
f'(\sigma_1) = \epsilon_2 - 2\alpha \sigma_1
$$

We consider two cases.

\begin{itemize}
\item {\bf Case I:} $\epsilon_1 \ge \frac{\epsilon_2}{2\alpha}$.

Setting $f'(\sigma_1) = 0$, we find that the function reaches a fixed
point at $\sigma_1 = \frac{\epsilon_2}{2\alpha}$.
At this point, $f''(\sigma_1) = -2\alpha$, which is negative. Hence 
$f(\sigma_1)$ is maximized at $\sigma_1 = \frac{\epsilon_2}{2\alpha}$.
We note that this value of $\sigma_1$ does not violate Constraint 1,
and hence this is a feasible solution.
In this case, the optimal settings are: 
$\sigma_1 = \frac{\epsilon_2}{2\alpha}$ and 
$\sigma_2 = \frac{\epsilon_2}{2}$. Thus 
$s_1 = \frac{2\alpha}{\epsilon}$ and 
$s_2 = \frac{2}{\epsilon_2}$.
The space complexity is 
$O(\frac{1}{\sigma_1 \sigma_2}) = O(\frac{4\alpha}{\epsilon_2^2})$.

\item
{\bf Case II:} $\epsilon_1 < \frac{\epsilon_2}{2\alpha}$

The function $f(\sigma_1)$ is increasing for $\sigma_1$ from $0$ to 
$\frac{\epsilon_2}{2\alpha}$. Hence this will be maximized at the
point $\sigma_1 = \epsilon_1$. Thus, in this case the optimal settings
are $\sigma_1 = \epsilon_1$, and 
$\sigma_2 = \epsilon_2 - \alpha \epsilon_1$.
Thus, $s_1 = \frac{1}{\epsilon_1}$, and 
$s_2 = \frac{1}{\epsilon_2 - \alpha \epsilon_1}$.
The space complexity is: 
$O(\frac{1}{\epsilon_1 (\epsilon_2 - \alpha \epsilon_1)})$.

We note that since $\epsilon_2 > 2\alpha\epsilon_1$, 
we have $(\epsilon_2 - \alpha \epsilon_1) > \frac{\epsilon_2}{2}$, and hence
the space complexity is $O(\frac{1}{\epsilon_1 \epsilon_2})$.
\end{itemize}

\end{proof}
%---------------------------------

\begin{lemma} 
The time taken to update the sketch on receiving each element of the
stream is $O(\max(s_1,s_2))$.
\end{lemma}

\begin{proof}
In processing an element
$(x,y)$ of the stream by Algorithm \ref{algo:update}, 
the following three scenarios may arise.
\begin{enumerate}
\item $x$ is present in $H$, and $y$ is present in $H_x$. We implemented the tables as hash tables,
hence the time taken to look up and increment ${\hat{f}}_x$ from $H$
and ${\hat{f}}_{x,y}$ from $H_x$ is O(1).
\item $x$ is present in $H$, but $y$ is not in $H_x$. If
the size of $H_x$ exceeds its space budget $s_2$, then, the time
taken to decrement the frequencies of all the stored values of the
secondary attribute is $O(s_2)$.
\item $x$ is not present in $H$. If the size of $H$ exceeds its space budget $s_1$,
then the time taken to decrement the frequencies of all the stored
values of the primary attribute is $O(s_1)$.
\end{enumerate}
The time complexity to update the sketch on receiving each element
is the maximum of these three, which establishes the claim.
\end{proof}

%---------------------
\section{Experiments}
\label{sec:simulation}
%---------------------
We simulated our algorithm for finding correlated heavy-hitters
in C++, using the APIs offered by the Standard Template Library
\cite{stl}, on three different datasets:
\begin{itemize}
\item {\bf IPPair:} An anonymized packet header trace collected by CAIDA
\cite{CAIDA} in both directions of an OC48 link. We used windump \cite{windump} in 
conjunction with a custom Java application to extract the source
IP address, the destination IP address, the source port number and
the destination port number from the .pcap files. Then, we took the comibation of (destination IP, source IP)
tuples to create this dataset. ``IPPair'' had 1.4 billion such tuples.
\item {\bf PortIP:} This is generated from the same trace as ``IPPair'', but it is a sample of 20.7 million (destination port,
destination IP) tuples.
\item {\bf NGram:} It is the ``English fiction'' 2-grams dataset based on the Google n-gram dataset \cite{ngram}. 
This is a collection of 2-grams extracted from books predominantly in the English language that a library or publisher identified as fiction. Some of the interesting trend analysis of 2-grams in English fiction can be found here \cite{trend}, e.g., the 2-gram ``child care'' started replacing the 2-gram ``nursery school'' in the mid-1970s. We took a uniform random sample of this dataset. We will refer to the two elements of a tuple as the ``first gram'' and the ``second gram'' respectively.
\end{itemize}

{\bf Objective: }The goal of the simulation was threefold: first, to learn about
typical frequency distributions along both the dimensions in real two-dimensional data streams; second, to illustrate the reduction in space and
time cost achievable by the small-space algorithm in practice; and
finally, to demonstrate how the space budget (and
hence, the allocated memory) influences the accuracy of our algorithm in practice.\\

For the {\em first} objective, we ran a naive algorithm on a smaller
sample of size 248 million taken from the ``IPPair'' dataset, where all
the distinct destination IPs were stored, and for each distinct
destination IP, all the distinct source IPs were stored. We
identified (exactly) the frequent values along both the dimensions
for ${\phi}_1 = 0.001$ and ${\phi}_2 = 0.001$. Only 43 of the 1.2 million
distinct destination IPs were reported as heavy-hitters. For the
secondary dimension, we ranked the heavy-hitter destination IPs
based on the number of distinct source IPs they co-occurred with,
and the number of distinct source IPs for the top eight are shown in
Figure \ref{fig:F_0_Second_Attribute}. All these heavy-hitter
destination IPs co-occurred with 9,000-18,000 {\em distinct} source
IPs, whereas, for all of them, the number of co-occurring {\em
heavy-hitter} source IPs was in the range 20-200 (note that the
Y-axis in Figure \ref{fig:F_0_Second_Attribute} is in log scale).
This shows that the distribution of the primary attribute values, as
well as that of the secondary attribute values for a given value of
the primary attribute, are very skewed, and
hence call for the design of small-space approximation algorithms like ours.

We did a similar exercise for the ``NGram'' dataset, and the result is in Figure \ref{fig:F_0_Second_Attribute_ngram}.
Once again, the values of ${\phi}_1$ and ${\phi}_2$ were both 0.001, and note that the number of distinct second
grams, co-occurring with the first grams, varies between 10 million and 100 million, but the number of CHH second grams
vary between 10 and 100 only, orders of magnitude lower than the number of distinct values of the second grams.

Since the ``NGram'' dataset is based on English fiction text, we observed some interesting patterns while working with the dataset: pairs of words that occur frequently together, as 
reported by this dataset, are indeed words whose co-occurrence intuitively look natural. We present some examples in Table \ref{tab:freq_word_pairs}, alongwith their frequencies:

\begin{table}[ht]
\caption{Pairs of words frequently occurring together}
\begin{tabular}{| l | l | l | l |}
\hline
{Gram1} & {Frequency of Gram1} & {Gram2} & {Frequency of Gram2 alongwith Gram1}\\
\hline
are & 1989774 & hardly & 4717\\
are & 1989774 & meant & 5031\\
still & 1601172 & remained & 4798\\
out & 1777906 & everything & 5497\\
was & 2373607 & present & 7932\\
was & 2373607 & deserted & 7641\\
look & 1226326 & outside & 2052\\
could & 1215055 & suggest & 5081\\
\hline
\end{tabular}
\label{tab:freq_word_pairs}
\end{table}

\begin{figure}[!ht]
\begin{center}
\includegraphics[width=3.5in]{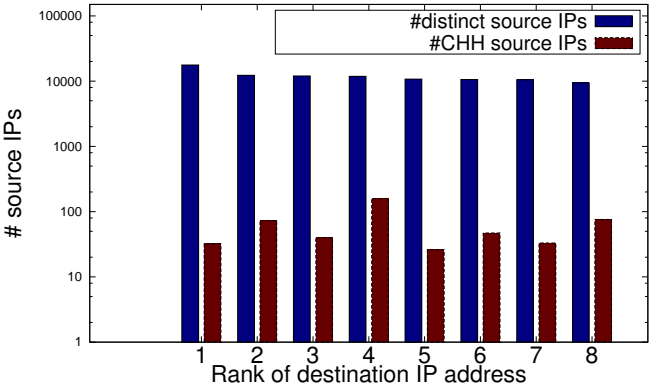}
\caption{\sf Basic statistics for a sample of ``IPPair''. On the X-axis are the ranks of the eight (heavy-hitter) destination IPs, that co-appear with maximum number of distinct
source IPs. For each destination IP, the Y-axis shows 1) the number
of distinct source IPs co-occurring with it, 2) the number of
heavy-hitter destination IPs co-appearing with it. Note that the Y-axis is logarithmic.}
\label{fig:F_0_Second_Attribute}
\end{center}
\end{figure}

\begin{figure}[!ht]
\begin{center}
\includegraphics[width=3.5in]{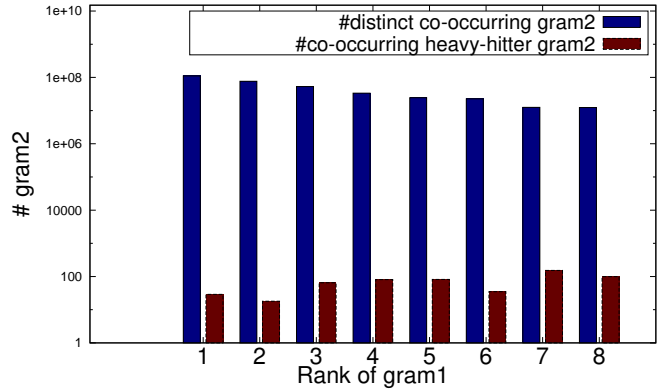}
\caption{\sf Basic statistics for ``NGram''. On the X-axis are the ranks of the eight (heavy-hitter) first gram values, that co-appear with maximum number of distinct second grams. For each first gram, the Y-axis shows 1) the number
of distinct second grams co-occurring with it, 2) the number of
heavy-hitter second grams co-appearing with it. Note that the Y-axis is logarithmic.}
\label{fig:F_0_Second_Attribute_ngram}
\end{center}
\end{figure}

The {\em second} objective was accomplished by comparing the space and
time costs of the naive algorithm as above (on the same sample of size 248 million 
taken from the ``IPPair'' dataset),
with those of the small-space algorithm, run with $s_1 = 3000$ and
$s_2 = 2000$ (Figure \ref{fig:spacetime}). We defined the space cost
as the distinct number of (dstIP, srcIP) tuples stored
$\left(\sum_{d}|H_d|\right)$, which is 34 times higher for the naive
algorithm compared to the small-space one. Also, the naive algorithm
took more than twice as much time to
run the small-space one.

\begin{figure}[!ht]
\begin{center}$
\begin{array}{cc}
\includegraphics[width=2.5in]{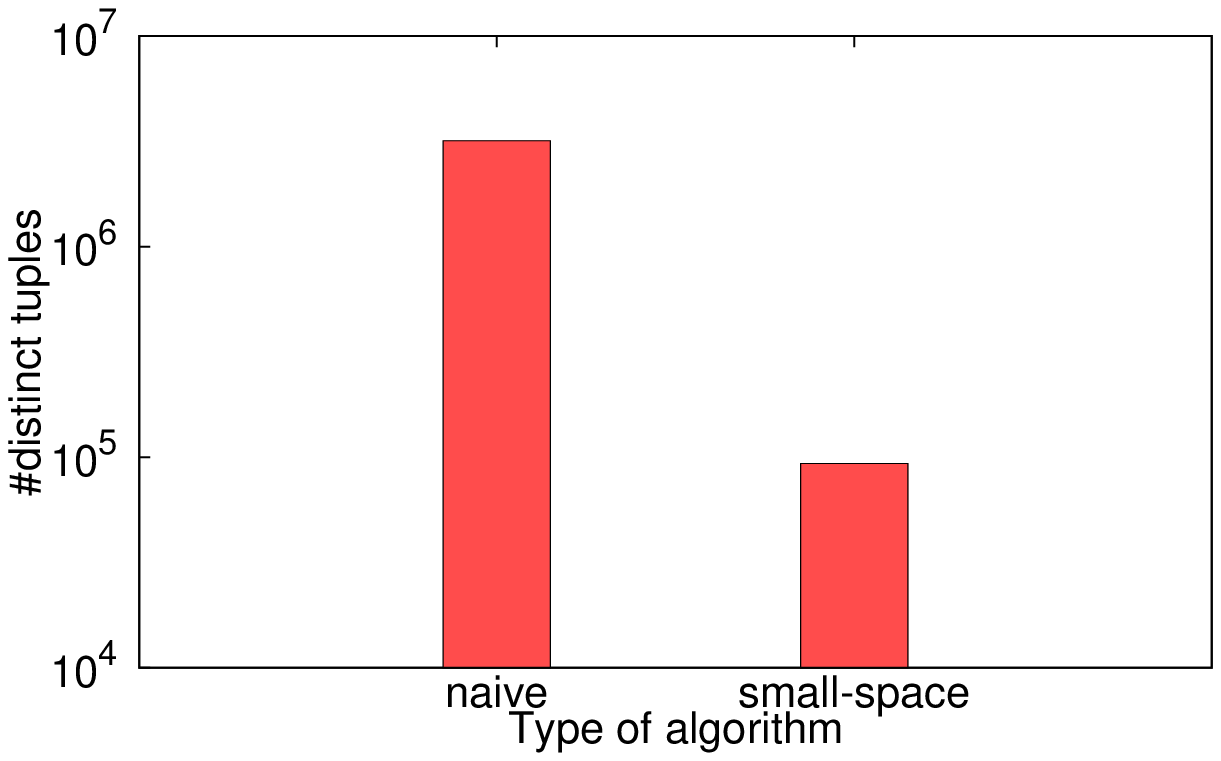}&
\includegraphics[width=2.5in]{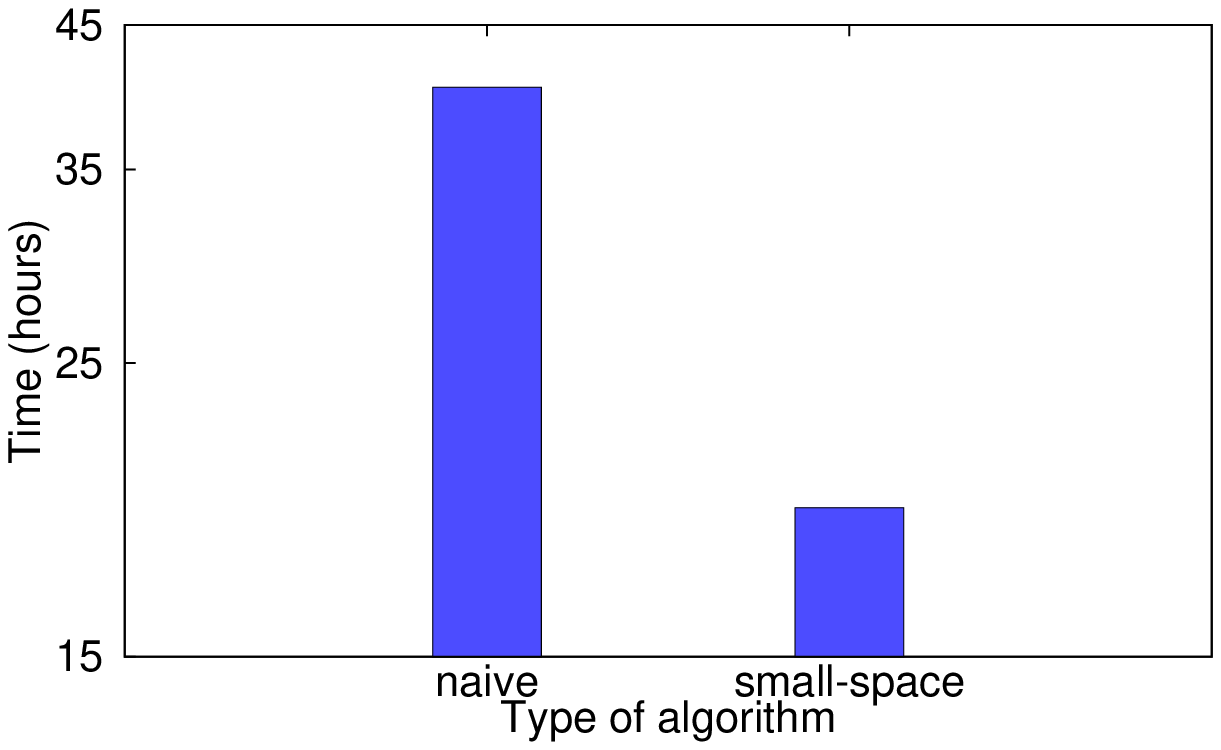}
\end{array}$
\caption{\sf Comparison of space (left) and time (right) costs of the naive and the small-space algorithms. The space is the total number of distinct tuples stored, summed over all distinct destination IP addresses. The time
is the number of hours to process the 248 million records. Note that
the Y-axis for the left graph is logarithmic.} \label{fig:spacetime}
\end{center}
\end{figure}

\begin{figure}[!ht]
\begin{center}
\includegraphics[width=3.5in]{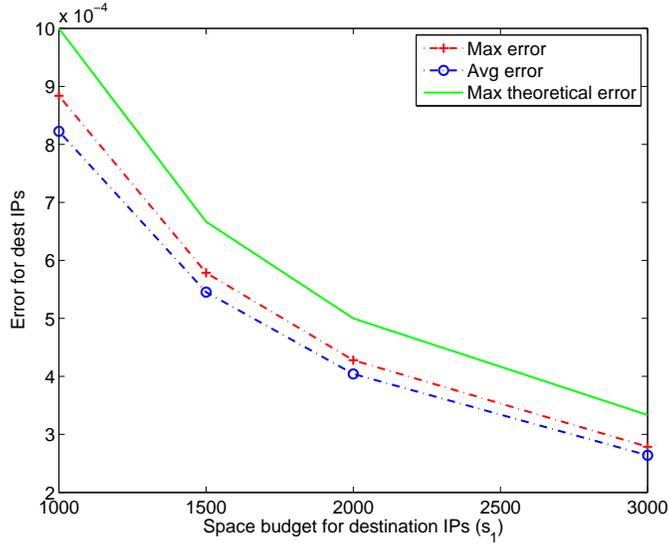}
\caption{{\sf Error statistic in estimating the frequencies of the
heavy-hitter destination IPs in ``IPPair''. The graph shows the theoretical
maximum ($\frac{1}{s_1}$), the experimental maximum and the
experimental average.}}
\label{fig:Error_First_Attribute}
\end{center}
\end{figure}

\begin{figure}[!ht]
\includegraphics[width=2.5in]{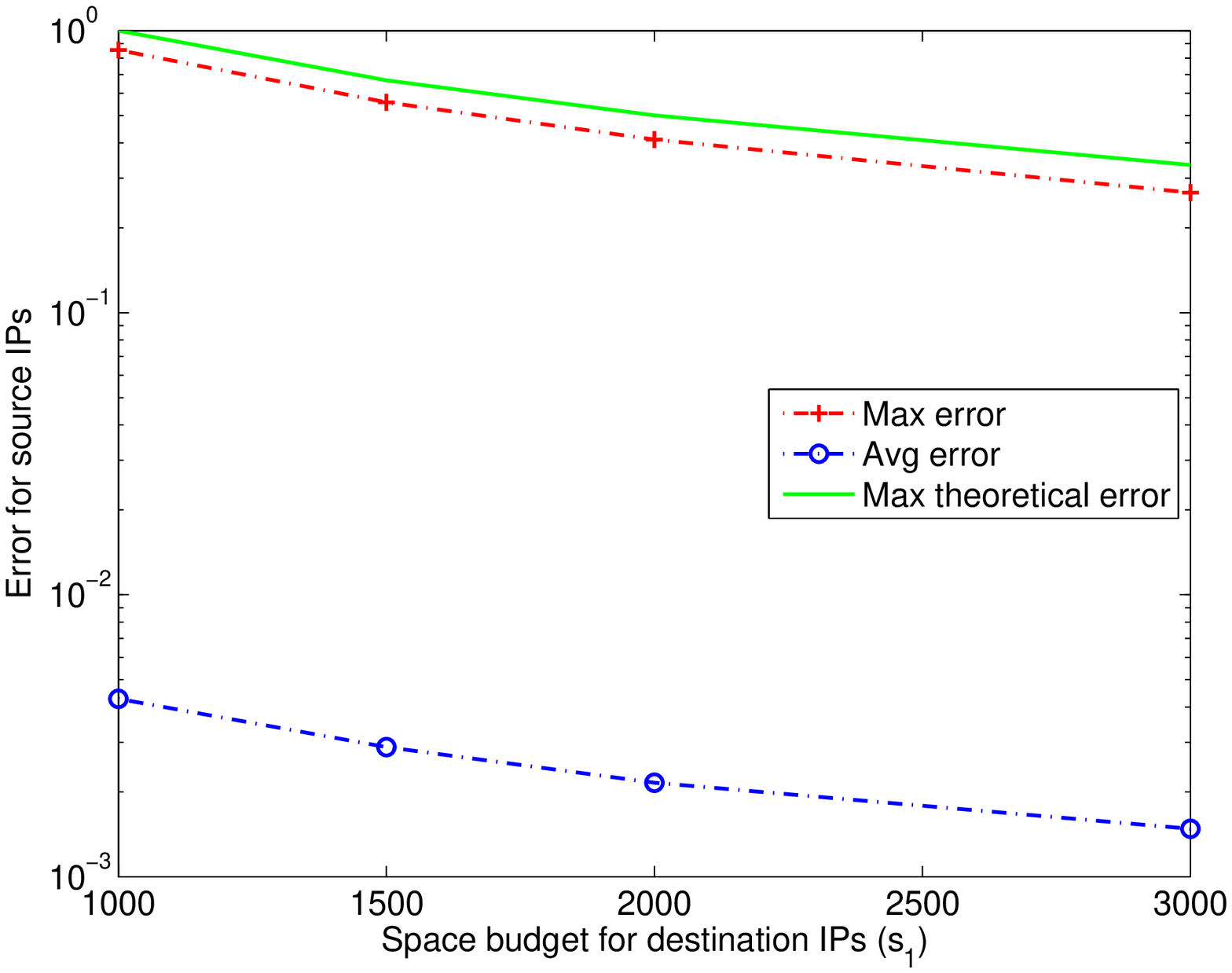}
\includegraphics[width=2.5in]{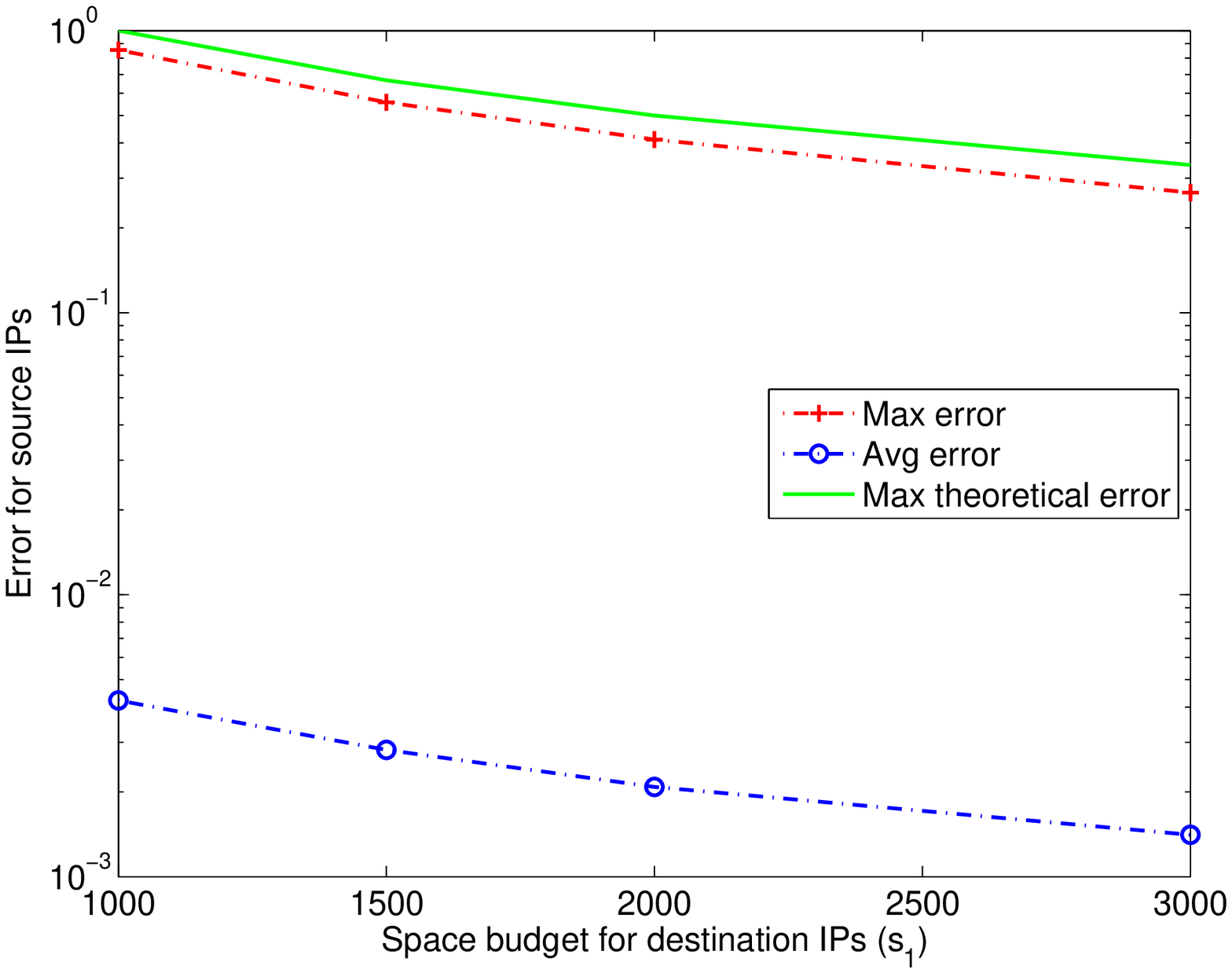}
\centering
\includegraphics[width=2.5in]{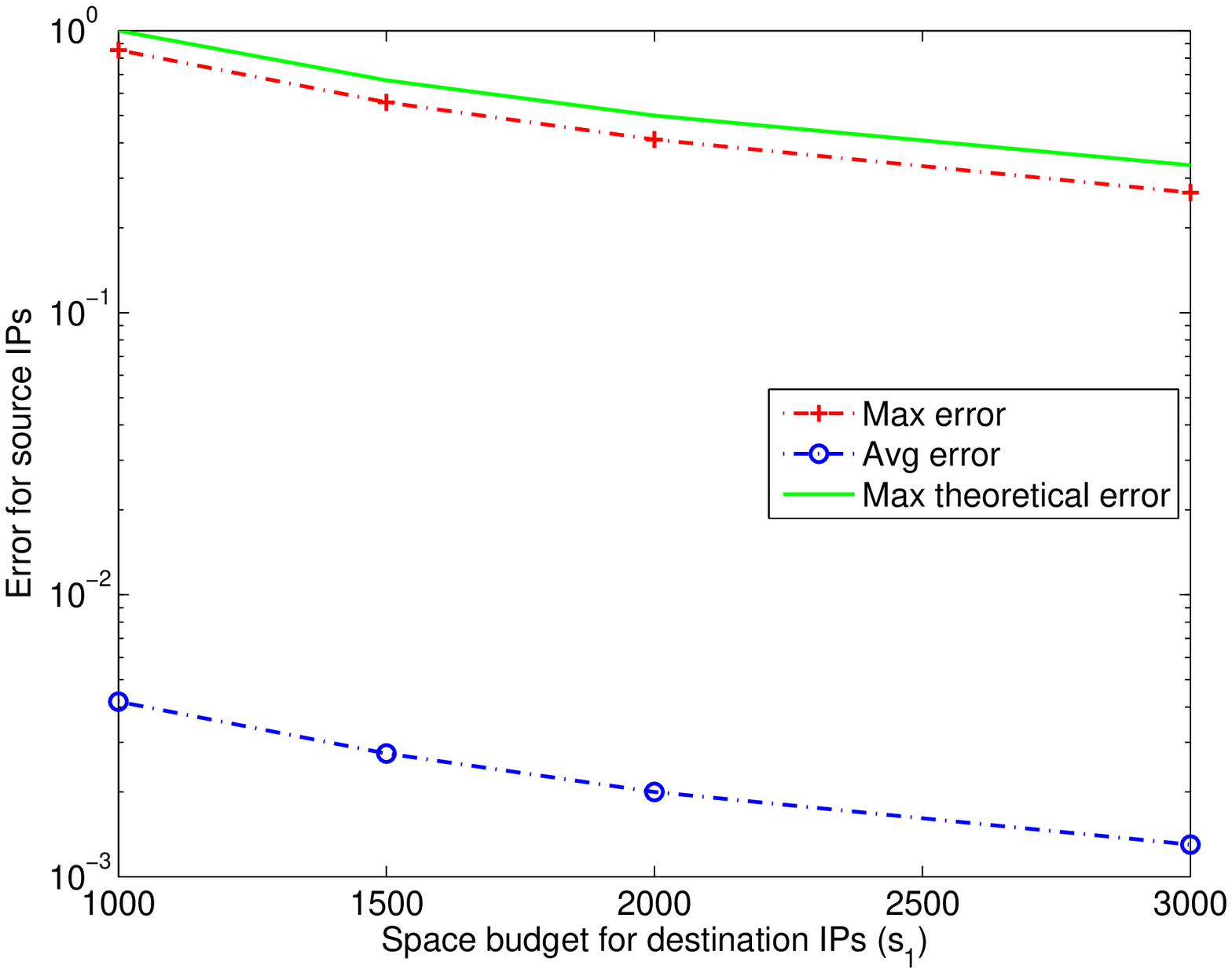}
\caption{{\sf Error statistic in estimating the frequencies of the
CHH source IPs in ``IPPair'', for $s_2$ = 1100, 1500 and 2000 respectively. The graph shows the
theoretical maximum $\left(\frac{1}{{\phi}s_1} + \frac{1}{s_2}\right)$, 
the experimental maximum and the experimental
average.}} \label{fig:Error_Second_Attribute}
\end{figure}

\begin{figure}[!ht]
\begin{center}
\includegraphics[width=3.5in]{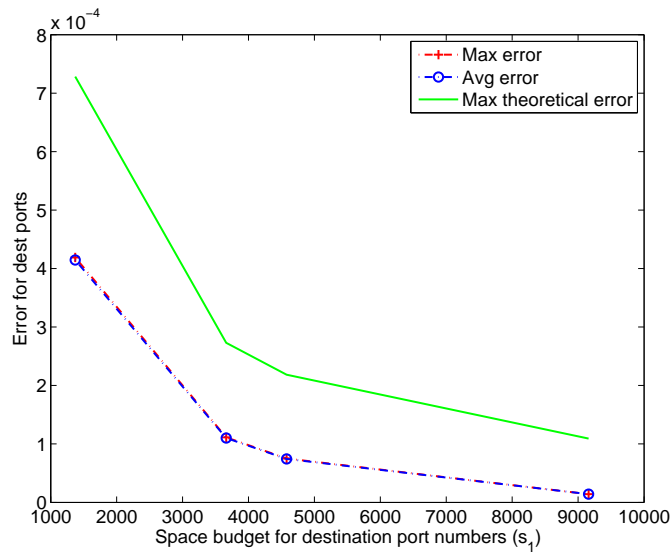}
\caption{{\sf Error statistic in estimating the frequencies of the
heavy-hitter destination ports from ``PortIP''}}
\label{fig:Error_First_Attribute_port}
\end{center}
\end{figure}

\begin{figure}[!ht]
\includegraphics[width=2.5in]{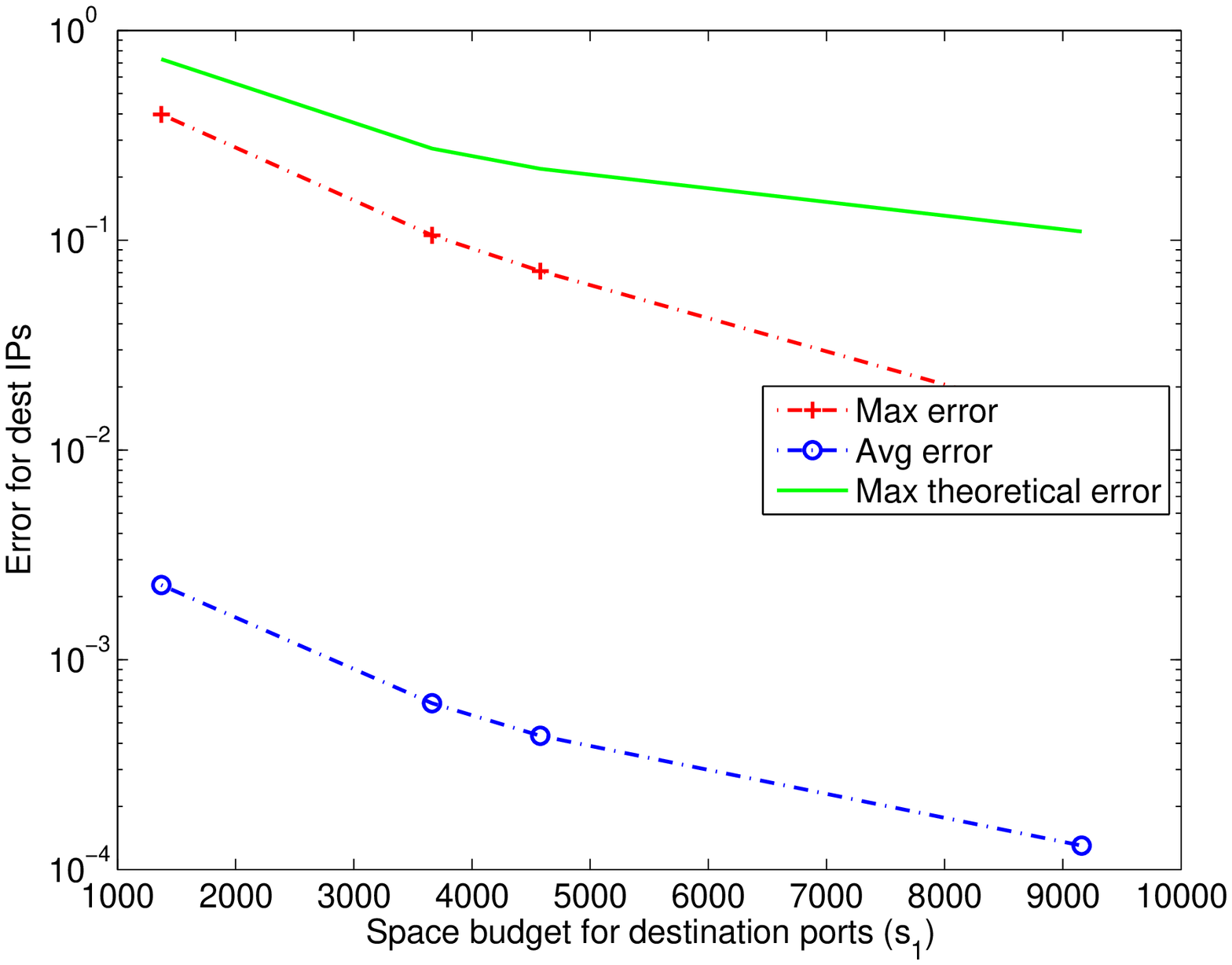}
\includegraphics[width=2.5in]{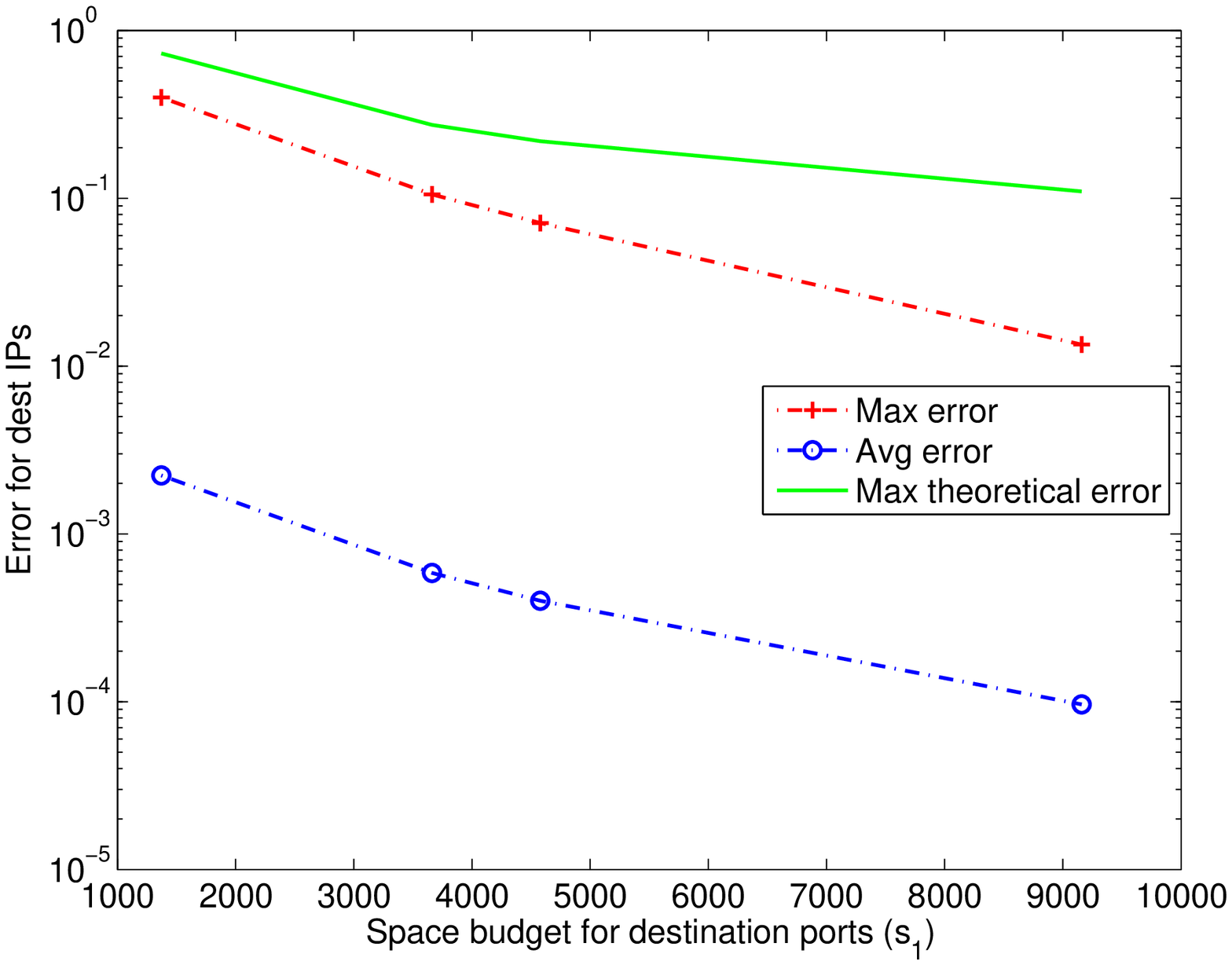}
\centering
\includegraphics[width=2.5in]{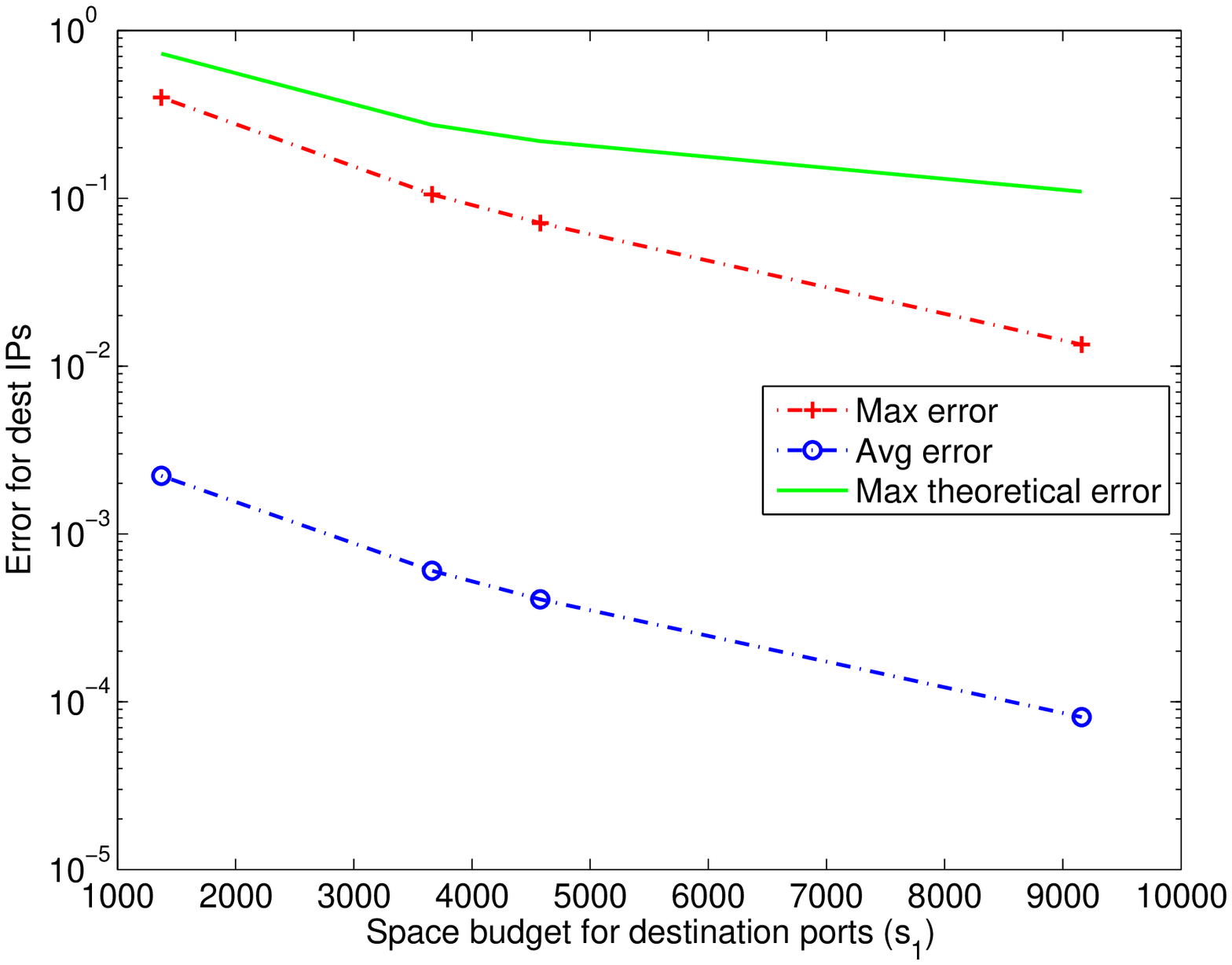}
\caption{{\sf Error statistic in estimating the frequencies of the
CHH destination IPs in ``PortIP''. The three graphs are for $s_2 =
1100$, $s_2 = 1500$ and $s_2 = 2000$ respectively.}}
\label{fig:Error_Second_Attribute_port}
\end{figure}

\begin{figure}[!ht]
\begin{center}
\includegraphics[width=3.5in]{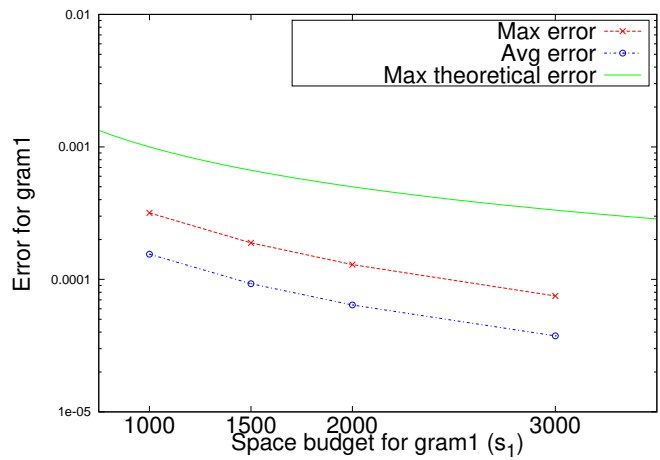}
\caption{{\sf Error statistic in estimating the frequencies of the
heavy-hitter first grams from ``NGram''}}
\label{fig:Error_First_Attribute_ngram}
\end{center}
\end{figure}

\begin{figure}[!ht]
\includegraphics[width=2.5in]{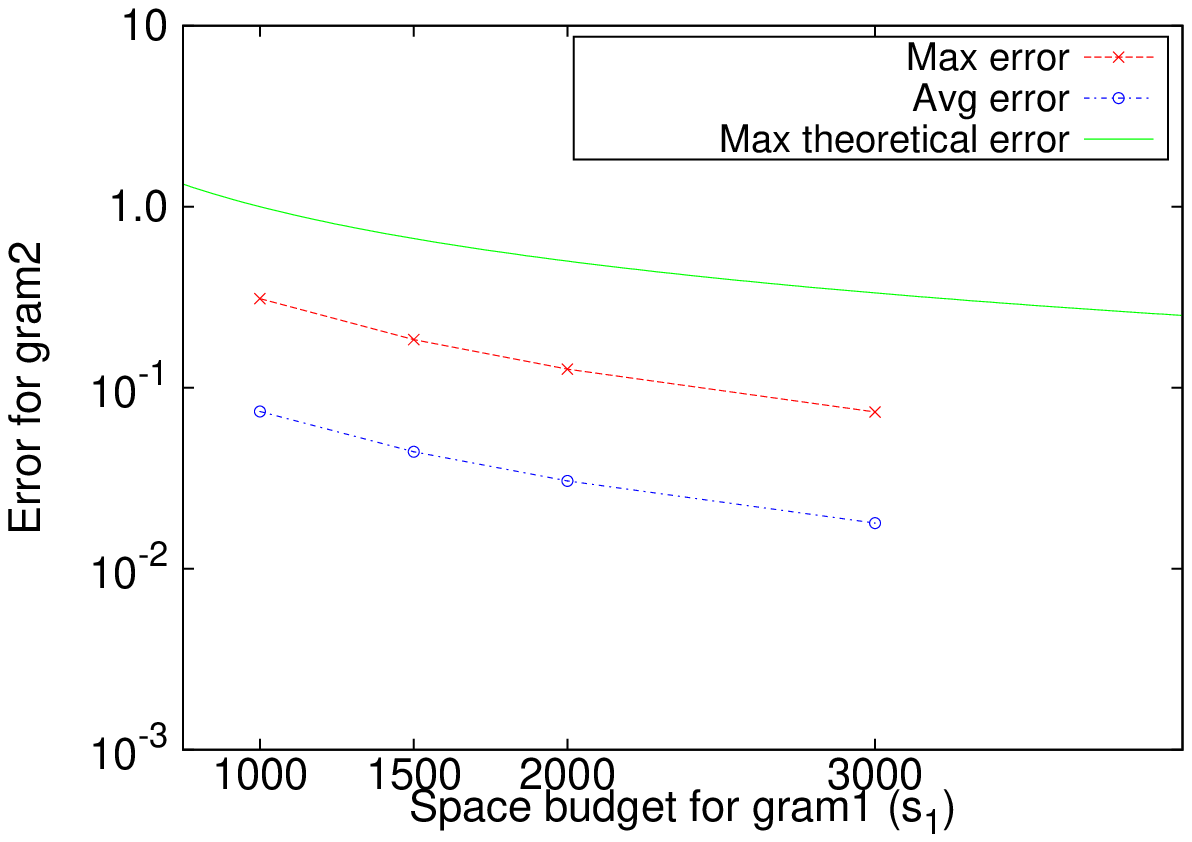}
\includegraphics[width=2.5in]{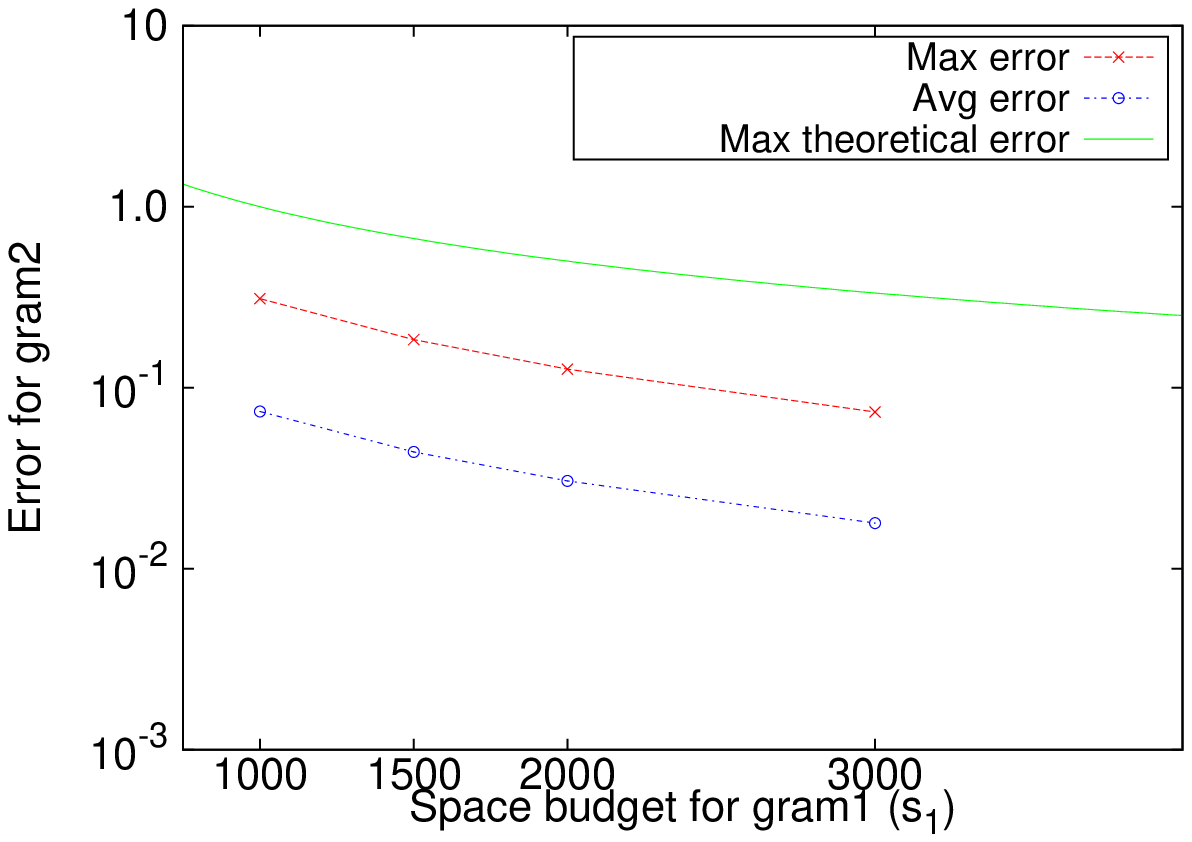}
\centering
\includegraphics[width=2.5in]{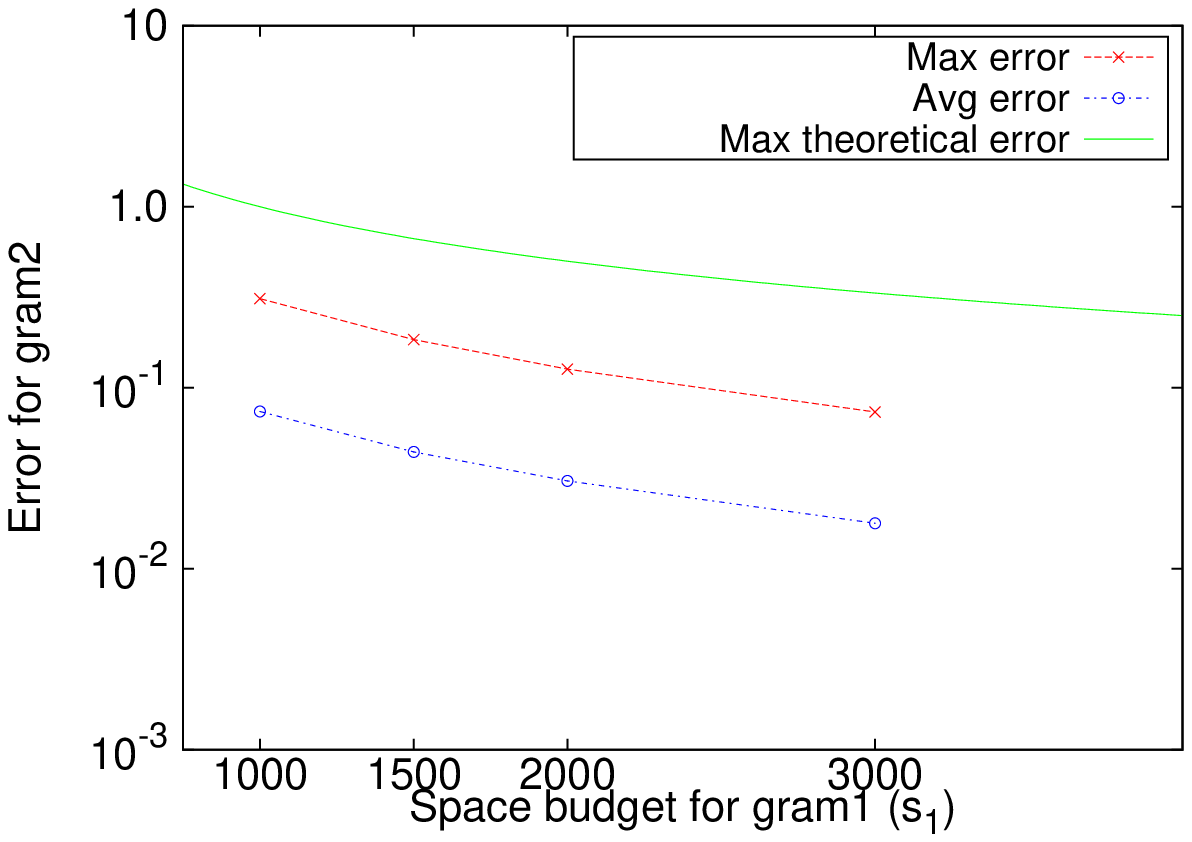}
\caption{{\sf Error statistic in estimating the frequencies of the
CHH second grams in ``NGram''. The three graphs are for $s_2 =
1100$, $s_2 = 1500$ and $s_2 = 2000$ respectively.}}
\label{fig:Error_Second_Attribute_ngram}
\end{figure}

For the {\em third} objective, we tested the small-space algorithm on all three datasets (with
different values of $s_1$ and $s_2$): ``IPPair'', ``PortIP'' and ``NGram''. To test the accuracy of our
small-space algorithm, we derived the ``ground truth'', i.e., a list
of the {\em actual} heavy-hitters along both the dimensions along
with their {\em exact} frequencies, by employing a four-pass
variant of the Misra-Gries algorithm (as discussed in Section \ref{sec:problem}).\\

{\bf Observations:} We define the error statistic in estimating the frequency of a
heavy-hitter value $d$ of the primary attribute as
$\frac{f_d-{\hat{f}}_d}{N}$, and in Figures
\ref{fig:Error_First_Attribute}, \ref{fig:Error_First_Attribute_port} and \ref{fig:Error_First_Attribute_ngram}, 
for each value of $s_1$, we plot
the maximum and the average of this error statistic over all the
heavy-hitter values of the primary attribute. We observed that both
the maximum and the average fell sharply as $s_1$ increased. Even by
using a space budget ($s_1$) as low as ~1000, the maximum error
statistic was only 0.09\% for ``IPPair'', 0.04\% for ``PortIP'' and 0.03\% for ``NGram''.\\

The graphs in Figures \ref{fig:Error_Second_Attribute}, \ref{fig:Error_Second_Attribute_port} and 
\ref{fig:Error_Second_Attribute_ngram} show the results of running our small-space algorithm with different values
of $s_1$ as well as $s_2$. We define the error statistic in
estimating the frequency of a CHH $s$ (that occurs
alongwith a heavy-hitter primary attribute $d$) as
$\frac{f_{d,s}-{\hat{f}}_{d,s}}{f_d}$, and for each combination of
$s_1$ and $s_2$, we plot the theoretical maximum, the experimental
maximum and the average of this error statistic over all CHH attributes. 
Here also, we observed that both the maximum and the average
fall sharply as $s_1$ increases. However, for a fixed value of
$s_1$, as we increased the value of $s_2$, the maximum did not
change at all (for either of three datasets), 
and the average did not reduce too much - this becomes evident if we compare the readings of the three sub-figures 
in Figures \ref{fig:Error_Second_Attribute}, \ref{fig:Error_Second_Attribute_port} and \ref{fig:Error_Second_Attribute_ngram}, which differ in their values of $s_2$, for identical values of $s_1$. The possible
reason is the number of CHHs being very low compared to
the number of distinct values of the secondary attribute occurring with a heavy-hitter
primary attribute, as we have pointed out in Figure
\ref{fig:F_0_Second_Attribute} for ``IPPair'' and Figure
\ref{fig:F_0_Second_Attribute_ngram} for ``NGram''. However, this is good because it
implies that in practice, setting $s_2$ as low as $\frac{1}{{\phi}_2}$
should be enough.

%---------------------
\section{Conclusion and Future Work}
\label{sec:conclusion}
%---------------------
For two-dimensional data streams, we presented a small-space
approximation algorithm to identify the heavy-hitters along the
secondary dimension from the substreams induced by the heavy-hitters
along the primary. We theoretically studied the relationship between
the maximum errors in the frequency estimates of the heavy-hitters
and the space budgets; computed the minimum space requirement along
the two dimensions for user-given error bounds; and tested our
algorithm to show the space-accuracy tradeoff for both the
dimensions.\\

Identifying the heavy-hitters along any one dimension allows us to
split the original stream into several important substreams; and
take a closer look at each one to identify the properties of the
heavy-hitters. In future, we plan to work on computing other
statistics of the heavy-hitters. For example, as we have already
discussed in Section \ref{sec:simulation}, our experiments with the
naive algorithm (on both the datasets) revealed that the number of
{\em distinct} secondary attribute values varied quite significantly
across the different (heavy-hitter) values of the primary attribute.
For any such data with high variance, estimating the variance in
small space \cite{BDMO03,ZG07} is an interesting problem in itself.
Moreover, for data with high variance, the simple arithmetic mean is
not an ideal central measure, so finding different quantiles, once
again in small space, can be another problem worth studying.

\bibliographystyle{plain}
\bibliography{references}

\end{document}